\documentclass[a4paper,UKenglish,cleveref, autoref, thm-restate,authorcolumns]{lipics-v2019}

\usepackage{todonotes}
\usepackage{comment}
\usepackage{bbm}
\usepackage[framemethod=tikz]{mdframed}
\DeclareMathOperator{\E}{\mathbb{E}}

\DeclareMathOperator{\OPT}{OPT}

\newtheorem{innercustomthm}{Lemma}
\newenvironment{customthm}[1]
  {\renewcommand\theinnercustomthm{#1}\innercustomthm}
  {\endinnercustomthm}

\renewcommand{\ge}{\geqslant}
\renewcommand{\le}{\leqslant}
\renewcommand{\geq}{\geqslant}
\renewcommand{\leq}{\leqslant}

\newcommand{\eps}{\varepsilon}

\DeclareMathOperator{\poly}{poly}
\DeclareMathOperator{\polylog}{polylog}

\newcommand{\Eopt}{E_{\text{\textsc{opt}}}}
\newcommand{\Enoise}{E_{\text{\textsc{noise}}}}
\newcommand{\Egood}{E_{\text{\textsc{good}}}}
\newcommand{\argmax}{\mathrm{argmax}}

\title{Robust Algorithms under Adversarial Injections\thanks{Research supported in part by the Swiss National Science Foundation project 200021-184656 ``Randomness in Problem Instances and Randomized Algorithms.''}} %TODO Please add

\titlerunning{Robust Algorithms under Adversarial Injections} %TODO optional, please use if title is longer than one line

\author{Paritosh Garg}{EPFL, Switzerland}{paritosh.garg@epfl.ch}{}{}%TODO mandatory, please use full name; only 1 author per \author macro; first two parameters are mandatory, other parameters can be empty. Please provide at least the name of the affiliation and the country. The full address is optional

\author{Sagar Kale}{University of Vienna, Austria}{sagar.kale@univie.ac.at}{}{}

\author{Lars Rohwedder}{EPFL,  Switzerland}{lars.rohwedder@epfl.ch}{}{}

\author{Ola Svensson}{EPFL,  Switzerland}{ola.svensson@epfl.ch}{}{}

\authorrunning{P. Garg, S. Kale, L. Rohwedder, and O. Svensson} %TODO mandatory. First: Use abbreviated first/middle names. Second (only in severe cases): Use first author plus 'et al.'

\Copyright{John Q. Public and Joan R. Public} %TODO mandatory, please use full first names. LIPIcs license is "CC-BY";  http://creativecommons.org/licenses/by/3.0/


 %TODO mandatory; please add comma-separated list of keywords

\category{} %optional, e.g. invited paper

\hideLIPIcs

\begin{document}

%\title{Robust Algorithms under Adversarial Injections} 
%\author{Paritosh Garg\thanks{
%	(paritosh.garg@epfl.ch)
%        EPFL, Lausanne, Switzerland.
%	}
%	\and Sagar Kale\thanks{
%        ()
%        University of Vienna.
%        }
%	\and Lars Rohwedder\thanks{
%        (lars.rohwedder@epfl.ch)
%        EPFL, Lausanne, Switzerland.
%        }
%	\and Ola Svensson\thanks{
%        (ola.svensson@epfl.ch)
%        EPFL, Lausanne, Switzerland.
%        }
%}

\nolinenumbers

\date{\today}

\maketitle

\begin{abstract}
% In this paper we study streaming and online algorithms for input in random order.
% In many problems a random order
% of the input sequence (as opposed to worst-case
% analysis) appears to be a necessary evil
% in order to prove satisfying guarantees.
% However, algorithmic techniques under this
% assumption are often vulnerable to even small
% disturbance, i.e., the random distribution is not perfectly uniform.
% 
% For this reason, we propose a new \emph{adversarial injections} model, in which an instance
% is ordered randomly, but an adversary may inject misleading elements at arbitrary positions.
% We believe that studying algorithms under 
% this much weaker assumption can lead to new insights and, in particular, to more robust algorithms.
% We consider two classical problems from combinatorial optimization in this model:
% Maximum matching and cardinality constrained monotone submodular function maximization.
% Our main technical contribution is a novel streaming algorithm for the latter, which computes
% a $0.55$-approximation. While the algorithm
% itself is clean and simple, an involved analysis
% shows that it emulates a subdivision of
% the input stream which can be used to greatly limit the power of the adversary.
% Building on the above.  --Sagar:
In this paper, we study streaming and online algorithms in the context of randomness in the input.
For several problems, a random order
of the input sequence---as opposed to the worst-case
order---appears to be a necessary evil
in order to prove satisfying guarantees.
However, algorithmic techniques that work under this
assumption tend to be vulnerable to even small 
changes in the distribution.
For this reason, we propose a new \emph{adversarial injections} model, in which the input
is ordered randomly, but an adversary may inject misleading elements at arbitrary positions.
We believe that studying algorithms under 
this much weaker assumption can lead to new insights and, in particular, more robust algorithms.
We investigate two classical combinatorial-optimization problems in this model:
Maximum matching and cardinality constrained monotone submodular function maximization.
Our main technical contribution is a novel streaming algorithm for the latter that computes
a $0.55$-approximation. While the algorithm
itself is clean and simple, an involved analysis
shows that it emulates a subdivision of
the input stream which can be used to greatly limit the power of the adversary.
\end{abstract}
\newpage

\section{Introduction}
\label{sec:introduction}
In the streaming model, an algorithm reads the input sequentially from the input stream while using limited memory. In particular, the algorithm is expected to use memory that is much smaller than the input size, ideally, linear in the size of a solution.  We consider the most fundamental setting in which the algorithm is further restricted to only   read the input stream once. % We further restrict the algorithms in this paper to only read the input stream once, which is the most fundamental setting.
In this case, the algorithm cannot remember much of the input along the way, and part of the input
is irrevocably lost.
Something similar happens for online algorithms:
Here, the input is given to the algorithm one element at a time and the algorithm
has to decide whether to take it into its solution or to discard it. This decision is 
irrevocable.

The most common approach to analyze the quality of an algorithm in these models is worst-case analysis. Here, an adversary has full knowledge of the algorithm's strategy
and presents a carefully crafted instance to it, trying to make the ratio between the value of the
algorithm's solution and that of an optimum solution (the approximation ratio; for online algorithms called
the competitive ratio) as small as possible\footnote{We assume that the problem is a maximization problem.}. While worst-case analysis gives very robust guarantees, it is also well-known that such an analysis is often very pessimistic.
Not only are good guarantees not possible
for many problems, but in many cases 
worst-case instances appear quite artificial. Hence, the worst-case approximation/competitive ratio does not necessarily represent the
quantity that we want to optimize.

One way to remedy this is to weaken the power of the adversary and a popular
model to achieve that is the random-order model.
Here, an adversary may pick the instance as before, but it is presented in a uniformly-random order to the algorithm.
This often allows for
significantly better provable guarantees.
A prime example is the secretary problem: For the worst-case order it is impossible to get a bounded competitive ratio whereas for the random-order a very simple stopping rule achieves a competitive ratio of $1/e$.
Unfortunately, in this model, algorithms tend to overfit and the assumption of
a uniformly-random permutation of the input is
a strong one.  
%To overcome these limitations, Kesselheim, Kleinberg, Niazadeh~\cite{KesselheimKN15} studied the robustness of strategies for the secretary problem when the order of the elements is not guaranteed to be uniformly-at-random but still contain a large amount of randomness (with respect to different measures such as entropy). A very recent and more adversary-centered model is the Byzantine model  by Bradac, Gupta, Singla, and Zuzic~\cite{bradac2019robust}. Their model is closely related to ours with some important differences that we discuss in Section~\ref{sec:rel_models} after defining our model in Section~\ref{sec:adv_inj_model}.
%To  we define it and give a more detailed  More recently and closer to the model that we introduce,    
%We will give examples of two common algorithmic
%
To illustrate this point, it is instructive to consider two examples of
techniques that break apart when the random-order assumption
is slightly weakened:
  
Several algorithms in the random-order model first read a small fraction of the input, say, the first $1\%$ of the input.
Such an algorithm relies on the assumption
that around $1\%$ of the elements from
an optimum solution are contained in this first chunk.
It computes some statistics, summaries, or initial solutions using this chunk in order to
estimate certain properties of the optimum solution.
Then in the remaining $99\%$ of the input
it uses this knowledge to build a good solution for the problem.
For examples of such streaming algorithms, see Norouzi-Fard et al.~\cite{norouzi2018beyond} who study submodular maximization and Gamlath et al.~\cite{gamlath2019weighted} who study maximum matching.  Also Guruganesh and Singla's~\cite{guruganesh2017online} online algorithm for maximum matching for bipartite graphs is of this kind. Note that these algorithms are very sensitive to noise at the beginning of the stream. 

Another common technique is to split the input into fixed parts
and exploit that with high probability the elements of the optimum solution are
distributed evenly among these parts,
e.g., each part has at most one optimum element.
These methods critically rely on the assumption that each part is representative for the
whole input or that the parts are in some way homogeneous
(properties of the parts are the same in expectation).  Examples of such algorithms include  the streaming algorithm for maximum matching~\cite{konrad2012maximum}, and the streaming algorithm for submodular maximization~\cite{agrawal2018submodular} that achieves the tight competitive ratio $1-1/e$ in the random-order model.

The motivation of this work is to understand whether the strong assumption  of uniformly-random order is necessary to
allow for better algorithms. More specifically, we are motivated by the following question:
\begin{center}
\begin{minipage}{0.9\textwidth}
\begin{mdframed}[hidealllines=true, backgroundcolor=gray!15]
Can we achieve the same guarantees as in the uniform-random order but by algorithms that are more robust against some distortions in the input?
\end{mdframed}
\end{minipage}
\end{center}
In the next subsection, we describe our proposed model that is defined so as to avoid overfitting to the random-order model, and, by working in this model, our algorithms for submodular maximization and maximum matching are more robust while maintaining good guarantees.

\subsection{The Adversarial Injections Model}\label{sec:adv_inj_model}
Our model---that we call the \emph{adversarial-injections} model---lies in between the two extremes of
random-order and adversarial-order.
In this model,
the input elements are divided into two sets
$\Enoise$ and $\Egood$.
An adversary first picks all elements, puts each element in either $\Enoise$ or $\Egood$, and chooses the input order.
Then the elements belonging to $\Egood$ are permuted
uniformly at random among themselves.
The algorithm does not know if an element is good or noise. We judge the quality of the solution produced by an algorithm by comparing it to the best solution in $\Egood$.

An equivalent description of the model is as follows. First, a set of elements is picked
by the adversary and is permuted randomly.
Then, the adversary injects more elements at positions of his choice
without knowing the random permutation of the original stream\footnote{We remark that the assumption that the adversary does not know the order of the elements is important. Otherwise, the model is equivalent to the adversarial order model for ``symmetric'' problems such as the matching problem. To see this, let  $\Eopt$ correspond to an optimum matching in any hard instance under the adversarial order.
Since a matching is symmetric, the adversary can inject appropriately renamed edges depending on the order of the edges (which he without this assumption knows) and obtain exactly the hard instance.
}.
Comparing with the previous definition, the elements injected by the adversary correspond to $\Enoise$ 
and the elements of the original stream correspond to $\Egood$.

We denote by $\Eopt \subseteq \Egood$ the elements of a fixed optimum solution of the elements in $\Egood$.  We can assume without loss of generality that $\Egood = \Eopt$, because otherwise elements in $\Egood\setminus\Eopt$ can be treated as those belonging to $\Enoise$ (which only strengthens the power of the adversary).

\subsection{Related Models}\label{sec:rel_models}
 With a similar motivation, Kesselheim, Kleinberg, Niazadeh~\cite{KesselheimKN15} studied the robustness of algorithms for the secretary problem from a slightly different perspective: They considered the case when the order of the elements is not guaranteed to be uniformly-at-random but still contains ``enough'' randomness with respect to different notions such as entropy. Recently, Esfandiari, Korula, Mirrokni~\cite{esfandiari2015online} introduced a model where the input is a combination of stochastic input that is picked from a distribution and adversarially ordered input. Our model is different in the sense that the input is a combination of randomly ordered (instead of stochastic input) and adversarially ordered elements.   

 Two models that are more similar to ours in the sense that the input is initially ordered in a uniformly-random order and then scrambled by an adversary in a limited way are~\cite{guha_mcg_06} and~\cite{bradac2019robust}. 
First, in the streaming model,  
 Guha and McGregor~\cite{guha_mcg_06} introduced the notion of a \mbox{$t$-bounded} adversary that can disturb a uniformly-random stream but has memory to remember and delay at most $t$ input elements at a time.  Second,   Bradac et al.~\cite{bradac2019robust} very recently introduced a new model  that they used to obtain robust online algorithms
for the secretary problem.  Their model, called the \emph{Byzantine} model, is very related to ours: the input is split into two sets which exactly correspond to $\Egood$ and $\Enoise$ in the adversarial-injections model.
The adversary gets to pick the elements in both of them, but
an algorithm will be compared against only $\Egood$.
Then---this is where our models differ---the adversary chooses an arrival time in $[0,1]$ for each element in $\Enoise$.
He has no control over the arrival times of the elements in $\Egood$, which
are chosen independently and uniformly at random in $[0,1]$.
The algorithm does not know to which set an element belongs, but it
knows the timestamp of each element, as the element arrives.
While the Byzantine model prevents certain kinds of overfitting (e.g., of the classical algorithm for the secretary problem),
it does not tackle the issues of the two algorithmic techniques we discussed earlier:
Indeed, by time $t = 0.01$, we will see around $1\%$ of the elements
from $\Eopt$. Hence, we can still compute some estimates based on them, but do not lose a lot
when dismissing them. Likewise, we may partition the timeline, and thereby the input, into parts such that in each part at most one element of $\Eopt$ appears.

Hence, even if our model appears very similar to the Byzantine model,
there is this subtle, yet crucial, difference.
The adversarial-injections model does not
add the additional dimension of time, and hence, does not allow for the kind of overfitting
that we discussed earlier. To further emphasize this difference, we now describe why it is strictly harder to devise algorithms in  the adversarial-injections model compared to  the Byzantine model.  It is at least as hard as the Byzantine model, because any algorithm for the former also works for the latter. This holds because 
the adversarial-injections model can be thought of as the Byzantine model with additional power to the adversary and reduced power for the algorithm:  The adversary gets the additional power of setting the timestamps of elements in $\Egood$, but not their identities, whereas the algorithm is not allowed to see the timestamp of any element.

To show that it is strictly harder, consider online bipartite matching.  We show that one cannot beat $1/2$  in the adversarial-injections model (for further details, see Section~\ref{sec:mat-online}) whereas we observe that the $(1/2+\delta)$-approximation algorithm~\cite{guruganesh2017online} for bipartite graphs and its analysis generalizes to the Byzantine model as well. 
This turns out to be the case because the algorithm in~\cite{guruganesh2017online} runs a greedy algorithm on the first small fraction, say $1\%$ of the input and ``augments'' this solution using the remaining $99\%$ of the input. The analysis crucially uses the fact that $99\%$ of the optimum elements are yet to arrive in the augmentation phase.  This can be simulated in the Byzantine model using timestamps in the online setting as one sees $1\%$ of $\Eopt$ in expectation.

\subsection{Our Results}
\label{sec:results}
%In order to provide evidence that our model can lead to new and interesting insights,
We consider two benchmark problems in combinatorial optimization under the adversarial-injections model in 
both the streaming and the online settings, namely
maximum matching and monotone submodular maximization subject to a cardinality constraint.
As we explain next, the study of these classic problems in our new model gives interesting insights: for many settings we can achieve more robust algorithms with similar guarantees as in the random-order model but, perhaps surprisingly, there are also natural settings where the adversarial-injection model turns out to be as hard as the adversarial order model. 

\textbf{The maximum matching problem.} We first discuss the (unweighted) \emph{maximum matching} problem.
Given a graph $G = (V,E)$, a matching $M$ is a subset of edges such that every vertex has at most one incident edge in $M$. A matching of maximum cardinality is called a maximum matching,
whereas a \emph{maximal} matching is one in which no edge can be added
without breaking the property of it being a matching.  The goal in the maximum matching problem is to compute a matching of maximum cardinality.  Note that a maximal matching is $1/2$-approximate.
Work on maximum matching has led to several important concepts and new techniques in theoretical computer science~\cite{mulmuley1987matching, lovasz1979determinants, edmonds1965paths, karp1986constructing}.
The combination of streaming and random-order model was first studied by Konrad, Magniez and Mathieu~\cite{konrad2012maximum}, where edges of the input graph arrive in the stream.
We allow a streaming algorithm to have memory $O(n \polylog(n))$,
which is called the semi-streaming setting. This is usually significantly less than
the input size, which can be as large as $O(n^2)$.
This memory usage is justified, because even storing a solution can
take $\Omega(n \log(n))$ space ($\Omega(\log(n))$ for each edge identity).
The question that Konrad et al.\ answered affirmatively was whether the
trivial $1/2$-approximation algorithm that computes a maximal matching can be improved in the random-order model.
Since then, there has been some work on improving the constant~\cite{gamlath2019online, farhadi2020approximate}.
The state-of-the-art is an approximation ratio of $6/11 \approx 0.545$ proved by
Farhadi, Hajiaghayi, Mah, Rao, and Rossi~\cite{farhadi2020approximate}.
We show that beating the ratio of $1/2$ is possible also in the adversarial-injections model by building on the techniques developed for the random-order model. 
\begin{theorem} \label{thm:matching}
There exists an absolute constant $\gamma>0$ such that there is a semi-streaming algorithm for
maximum matching under adversarial-injections with an approximation ratio of $1/2 + \gamma$
in expectation.  
\end{theorem}
We note that beating $1/2$ in adversarial-order streams is a major open problem. In this regard, our algorithm can be viewed as a natural first step towards understanding this question.

Now we move our attention to the online setting, where the maximum matching problem was first studied in the seminal work of Karp, Vazirani, and Vazirani~\cite{karp1990optimal}. They gave a tight $(1 - 1/e)$-competitive algorithm for the so-called one-sided vertex arrival model which is an important special case of the edge-arrival model considered here. Since then, the online matching  problem has received significant attention (see e.g.~\cite{buchbinder2019online, epstein2012improved, feige2018tighter, huang2019tight, gamlath2019online}).    Unlike the adversarial streaming setting, there is a recent hardness result due to \cite{gamlath2019online} in the adversarial online setting that the trivial ratio of $1/2$ cannot be improved. We also know by \cite{guruganesh2017online} that one can beat $1/2$ for bipartite graphs in the random-order online setting. Hence, one might hope at least for bipartite graphs to use existing techniques to beat $1/2$ in the online adversarial-injections setting and get a result analogous to Theorem~\ref{thm:matching}.  But surprisingly so, this is not the case. 
We observe that the construction used in proving Theorem~3 in \cite{gamlath2019online} also implies
that there does not exist an algorithm with a competitive ratio of $1/2 + \eps$ for any $\eps > 0$
in the adversarial-injections model.\\[-1mm]

\textbf{Maximizing a monotone submodular function subject to a cardinality constraint.} 
In this problem, we are given a ground set $E$ of $n$ elements and a monotone submodular
set function $f : 2^E \rightarrow \mathbb R_{\ge 0}$. A function is said to be submodular, if for
any $S, T\subseteq E$ it holds that $f(S) + f(T) \ge f(S\cup T) + f(S\cap T)$. It is monotone
if $f(S) \le f(T)$ for all $S\subseteq T \subseteq E$.
The problem we consider is to find a set $S\subseteq E$ with $|S|\le k$ that maximizes $f(S)$.
We assume that access to $f$ is via an oracle.
% An important and easy to grasp special case is that of maximum coverage:
% Here we are given sets $S_1,\dotsc,S_n \subseteq U$,
% and we want to select $k$ of them to maximize the size of their union.

In the offline setting, a simple greedy algorithm that iteratively picks the element with the largest marginal contribution to $f$ with respect to the current solution is
$(1 - 1/e)$-approximate~\cite{nemhauser1978analysis}.  
%Unfortunately, this algorithm does not work
%in the streaming setting, since it would need
%$k$ passes over the stream.
This is tight: Any algorithm that achieves an approximation ratio of better than $(1 - 1/e)$ must make $\Omega(n^k)$ oracle calls~\cite{nemhauser1978best}, which is enough to brute-force over all $k$-size subsets. Even for  maximum coverage (which is a special family of monotone submodular functions), it is NP-hard to get an approximation algorithm with ratio better than $1-1/e$~\cite{feige}.

In the random-order online setting, this problem is called the \emph{submodular secretary} problem, and an exponential time $1/e$-approximation and polynomial-time $(1 - 1/e)/e$-approximation algorithms are the state-of-the-art~\cite{kesselheim_et_al17}.  In the adversarial online setting, it is impossible to get any bounded approximation ratio for even the very special case of picking a maximum weight element.  In this case,  $|\Eopt| = 1$ and adversarial and adversarial-injections models coincide; hence the same hardness holds.  In light of this negative result, we focus on adversarial-injections in the streaming setting.
Note that to store a solution we only need
the space for $k$ element identities.
We think of $k$ to be much smaller than $n$.
Hence, it is natural to ask,
whether the number of elements in memory can
be independent of $n$.

For streaming algorithms in the adversarial order
setting, the problem was first studied by Chakrabarti and Kale~\cite{chakrabarti2015submodular} where they gave a $1/4$-approximation algorithm.
This was subsequently improved to $1/2 - \eps$  by Badanidiyuru et al.~\cite{badanidiyuru2014streaming}.
Later, Norouzi-Fard et al.~\cite{norouzi2018beyond} observed that in the random-order model this ratio can be
improved to beyond $1/2$.
Finally, Agrawal et al.~\cite{agrawal2018submodular} obtained a tight $(1-1/e)$-approximation guarantee in the random-order model.

The algorithm of Agrawal et al.~\cite{agrawal2018submodular} involves as a crucial step a partitioning the stream in order to isolate the elements of the optimum solution. As discussed earlier, this approach does not work under adversarial-injections. 
However, we note that the algorithm and analysis by Norouzi-Fard et al.~\cite{norouzi2018beyond} can be easily modified to work under adversarial-injections as well.
Their algorithm, however, has an approximation ratio of $1/2 + 8 \cdot 10^{-14}$. In this paper, we remedy this weak guarantee. 

\begin{theorem} \label{thm:submod}
There exists a $0.55$-approximation algorithm that stores a number of elements that is independent of $n$ for maximizing a monotone submodular function with a cardinality constraint $k$ under adversarial-injections in the streaming setting.
\end{theorem}
We summarize and compare our results with random-order and adversarial-order models for the problems we study in Table~\ref{tab:one}. It is interesting to see that in terms of beating $1/2$, our model in the streaming setting agrees with the random-order model and in the online setting agrees with the adversarial-order model. 

\begin{table}
\begin{center} 
\caption{\label{tab:one}: Comparison of different models for the two studied problems. Here, $\gamma > 0$ is a fixed absolute constant and $\eps>0$ is any constant. }
\begin{tabular}{|l|l|l|l|}
 \multicolumn{4}{l}{\textbf{Maximum matching}} \\
 \hline
  & Random order & Adversarial Injections & Adversarial order \\ 
 \hline
 Streaming & $\geq 6/11$ \cite{farhadi2020approximate} & $\geq 1/2 + \gamma$ & $\leq 1 - 1/e + \eps$ \cite{kapralov2013better} \\
 \hline
 Online & $\geq 1/2$ (folklore) & $\leq 1/2$ & $\leq 1/2$ \cite{gamlath2019online}  \\
 \hline
 \multicolumn{4}{l}{\textbf{Submodular function maximization}} \\
 \hline
  & Random order & Adversarial Injections & Adversarial order \\ 
 \hline
  Streaming  & $\geq 1 - 1/e - \eps$ \cite{agrawal2018submodular} & $\geq 0.55$ & $\geq 1/2 - \eps$ \cite{badanidiyuru2014streaming} \\
  &  $\leq 1 - 1/e + \eps$ \cite{mcgregor2019better} & &  $\leq 1/2$ \cite{feldman_etal20}\\
 \hline
\end{tabular}
\end{center}
\end{table}

\section{Matching}
\label{sec:matching}
In this section, we consider the problem of
maximum unweighted matching
under adversarial injections in both streaming and online settings where
the edges of the input graph arrive one after another.

\subsection{Streaming Setting}
We show that the trivial approximation ratio of $1/2$
can be improved upon.
We provide a robust version of existing techniques and prove a statement about robustness of the greedy algorithm to achieve this.  

First, let us introduce some notation which we will use throughout this section.
We denote the input graph by $G = (V, E)$, and let $M^*$ be a maximum matching.
For any matching $M$, the union $M\cup M^*$ is
a collection of vertex-disjoint paths and cycles.
When $M$ is clear from the context,
a path of length $i\ge 3$ in $M\cup M^*$
which starts and ends with
an edge of $M^*$ is called an $i$-augmenting path.
Notice that an $i$-augmenting path alternates between
edges of $M^*$ and $M$ and that we can increase the size of $M$ by one by taking all edges from $M^*$ and removing
all edges from $M$ along this path.
We say that an edge in $M$ is $3$-augmentable
if it belongs to some $3$-augmenting path.
Otherwise, we say it is non-$3$-augmentable.  Also, let $M^* = \Eopt$; as described in the introduction, this is without loss of generality.

As a subroutine for our algorithm we need
the following procedure.
\begin{lemma}[Lemma 3.1 in \cite{gamlath2019weighted}]\label{find-3-aug}
There exists a streaming algorithm \textsc{3-Aug-Paths} with the following properties: 
\begin{enumerate}
    \item The algorithm is initialized with a matching $M$ and a parameter $\beta > 0$.
    Then a set $E$ of edges is given to the algorithm one edge at a time. 
    \item If $M \cup E$ contains at least $\beta|M|$ vertex disjoint 3-augmenting paths, the algorithm returns a set $A$ of at least $(\beta^2/32)|M|$ vertex disjoint 3-augmenting paths. The algorithm uses space $O(|M|)$. 
\end{enumerate}
\end{lemma}

\subsubsection{The Algorithm} \label{alg:mat}
We now describe our algorithm \textsc{Match}. It runs two algorithms in parallel and selects the better of the two outputs.
The first algorithm simply constructs a maximal matching greedily by updating the variable $M_1$.
The second algorithm also constructs a matching $M^{(1)}_2$ greedily, but it stops once $M^{(1)}_2$ has $|M^*|(1/2-\eps)$ edges. We call this Phase~1.
Then, it finds 3-augmentations using
the \textsc{3-Aug-Paths} algorithm given by Lemma~\ref{find-3-aug}. 
Finally, it augments the paths found to obtain a matching $M_2$. The constant
$\beta$ used in \textsc{3-Aug-Paths}
is optimized for the
analysis and will be specified there.

Notice that here we assumed that the algorithm knows $|M^*|$. 
This assumption can be removed using geometric guessing at a loss of an arbitrary small factor in the approximation ratio. We refer the reader to Appendix~\ref{app:matchopt}  
%\textcolor{red}{arxiv version(cite stuff)} 
for details. 

\subsubsection{Overview of the Analysis} \label{sec:mat-overview}
%arxiv version
%We discuss only the intuition here and defer the formal proof to %Appendix~\ref{app:match} due to space considerations.
We discuss only the intuition here and refer the reader to Appendix~\ref{app:match} 
%\textcolor{red}{the arxiv version(cite stuff)} 
for a formal proof.
Consider the first portion of the stream until we have seen a
small constant fraction of the elements in $\Eopt$.
If the greedy matching up to this point is already close to a $1/2$-approximation, 
this is good for the second algorithm as we are able to augment the matching using the remaining edges of $M^*$.
The other case is good for the first algorithm:
We will show that the greedy matching formed so far must contain a significant fraction of the edges in $M^*$ which we have seen so far.
If this happens, the first algorithm
outputs a matching of size a constant fraction more than $|M^*|/2$.

A technical challenge and novelty comes from the fact that
the two events above are
not independent of the random order of $\Eopt$.
Hence, when conditioning on one event,
we can no longer assume that the order of $\Eopt$ is uniformly at random.
We get around this by showing that the greedy algorithm
is robust to small changes in streams.
The intuition is that in the first part of the stream
the greedy solution either is large for all permutations
of $\Eopt$ or it is small for all permutations.
Hence, these are not random events depending on the
order, but two cases in which we can assume a uniform distribution.

\subsection{Online Setting}\label{sec:mat-online}
Since we can improve $1/2$ for the streaming setting,
it is natural to hope that the existing techniques
(e.g., the approach of the previous subsection)
can be applied in the online setting as well.
Surprisingly, this is not the case. In other words,
the competitive ratio of $1/2$ is optimal even for bipartite graphs.
The technique from the previous subsection breaks apart,
because the algorithm constructs several
candidate solutions in parallel by guessing $|M^*|$. 
This is not a problem for a streaming algorithm, however,
an online algorithm can only build one solution.

For a formal proof, we rely on the bipartite construction used in the proof of Theorem~3 from \cite{gamlath2019online}. The authors show that there is no 
(randomized) algorithm with a competitive ratio of $1/2 + \eps$ for any $\eps > 0$. More precisely,
they show that not even a good fractional matching can be
constructed online. For fractional matchings, randomization
does not help and therefore we can assume the algorithm is
deterministic.
The original proof is with respect to adversarial order,
but it is not hard to see that it transfers to
adversarial injections.

The authors construct a  bipartite instance that arrives in (up to) $N$ rounds. In round $i$, a matching of size $i$ arrives.
The algorithm does not know whether the current round is the last one or not.
Hence, it has to maintain a good approximation after each round.
This forces the algorithm to take edges that do not belong to the optimal matching and
eventually leads to a competitive ratio of $1/2$.
The same construction works in our model:
The edges from the optimal matching arrive in the last round and their internal order does not
affect the proof. In fact, the construction works for any order of the elements within a round.
Thus, an algorithm cannot exploit the fact that their order is randomized and therefore
also cannot do better than $1/2$.

%%% Local Variables:
%%% mode: latex
%%% TeX-master: "main_ai"
%%% End:

\section{Submodular Maximization}
\label{sec:submod}
In this section, we consider the problem of submodular maximization subject to a cardinality constraint.
The algorithm has query access to a monotone, submodular function
$f: 2^{E} \rightarrow \mathbb R$ over a ground set $E$.
Moreover, $f$ is normalized with $f(\emptyset) = 0$.
The goal is to compute a set $S$ of size at most $k$ that
maximizes $f(S)$.
We present a $0.55$-approximate streaming algorithm
in the adversarial-injections model which only
needs the memory to store $(O(k))^{k}$ many elements.
In particular, this number is independent of
the length of the stream.

%elements\footnote{To be precise, the approximation guarantee is $0.5506-O(\eps+\gamma+\delta)$ with number of elements stored $O((k/\eps)^k \cdot \log_{1+\gamma} \frac{k}{\delta})$ where $\eps$ is a function value discretization parameter, $\delta$ and $\gamma$ are constants related to guessing the value of the optimal solution. }. 

\subsection{Notation}
For $e \in E$ and $S\subseteq E$
we write $S+e$ for the set $S\cup\{e\}$ and
$f(e \mid S)$ for  $f(S+e) - f(S)$.
Similarly, for $A, B \subseteq E$ let
$f(A\mid B) := f(A\cup B) - f(B)$.
An equivalent definition of submodularity to the one given
in the introduction states that
for any two sets $S \subseteq T \subseteq E$, and
$e \in E\setminus T$ it holds that
$f(e \mid S) \geq f(e \mid T)$.

We denote by $\sigma$ the stream of elements $E$,
by $-\infty$ and $\infty$ the start and end of the stream.
For elements $a$ and $b$, we write $\sigma[a,b]$
for the interval including $a$ and $b$ and
$\sigma(a,b)$ for the interval excluding them.
Moreover, we may assume that $f(\emptyset) = 0$,
since otherwise, we may replace
the submodular function by $f': 2^E\rightarrow \mathbb R_{\ge 0}, T \mapsto f(T) - f(\emptyset)$.

%Recall, in our model $\Eopt$ is an optimal solution.
Denote the permutation of $\Eopt$ by $\pi$.
%Let $\pi^{(i)} = \{ \pi_1, \ldots , \pi_i \}$ for $i \in [k]$. 
Let $o^{\pi}_i$ be the $i$'th element of $\Eopt$ in the stream
according to the order given by $\pi$.
Let $O^{\pi}_0=\emptyset$ and
$O^{\pi}_i = \{o^\pi_1,\dotsc,o^\pi_i\}$ for all $i$; hence, $\Eopt = O^\pi_k$ for any $\pi$.
Finally, let $\OPT = f(O^\pi_k)$.

\subsection{The Algorithm} \label{alg:submod}
For simplicity we present an algorithm with the assumption
that it knows the value $\OPT$.
Moreover, for the set of increases in $f$,
that is $I = \{f(e \mid S) : e \in E, S \subseteq E\}$,
we assume that $|I| \le O(k)$.
These two assumptions can be made at a marginal cost in
the approximation ratio and an insignificant increase in memory.
This follows from standard techniques.
We refer the reader to 
Appendix~\ref{app:two} and Appendix~\ref{app:three}  
%\textcolor{red}{appendix of the arxiv version(cite stuff)} 
for details.

As a central data-structure, the algorithm maintains
a rooted tree $T$ of height at most $k$.
Every node except for the root stores a single element from $E$.
The structure resembles a prefix tree:
Each node is associated with the solution, where the elements 
on the path from the root to it is selected.
The nodes can have at most $|I|$ children, that is, one for
each increase.
The basic idea is that for some partial solution $S\subseteq E$ (corresponding to a node) and two elements $e, e'$
with $f(e\mid S) = f(e'\mid S)$ we only consider one of the solutions
$S\cup\{e\}$ and $S\cup\{e'\}$.
More precisely, the algorithm starts with a tree consisting
only of the root.
When it reads an element $e$ from the stream, it
adds $e$ as a child to every node where
(1) the distance of the node to the root is smaller than $k$ and
(2) the node does not have a child with increase $f(e \mid S)$, where $S$ is the partial solution corresponding to this particular node.

Because of (1), the solutions are always of cardinality at most $k$.
When the stream is read completely, the algorithm selects the
best solution among all leaves. An example of the algorithm's
behavior is given in Figure~\ref{fig:subm_example}.

%For any node $N$, denote by $S(N)$ the union of nodes encountered while traversing from root to $N$. For each possible $f$ increase $x \in I$, denote by $C(N,x)$
%the current set of children $N'$ such that $f(N'|S(N)) = x$. Whenever a new element $e$ arrives in the stream,
%for each node $N$ at height less than $k$, if $|C(N, f(e\mid S(N))| < 1$, then
%create a new child $\{e\}$.
%At the end of the stream, output $S(N)$ of size $k$ where $N$ is a node at height $k$ that maximizes $f$.
%
%For an example, see Figure~\ref{fig:subm_example}.

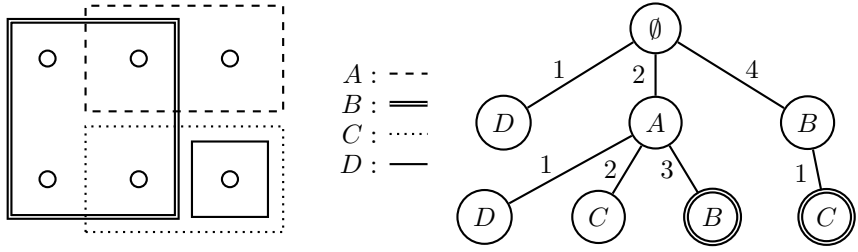
\begin{figure}
\centering
\begin{tikzpicture}
  \draw[thick] (0, 0) circle (3pt);
  \draw[thick] (1.2, 0) circle (3pt);
  \draw[thick] (2.4, 0) circle (3pt);
  \draw[thick] (0, 1.6) circle (3pt);
  \draw[thick] (1.2, 1.6) circle (3pt);
  \draw[thick] (2.4, 1.6) circle (3pt);
  \draw[thick, double] (-0.5, -0.5) rectangle (1.7, 2.1);
  \draw[thick, dotted] (0.5, -0.7) rectangle (3.1, 0.7);
  \draw[thick, dashed] (0.5, 0.9) rectangle (3.1, 2.3);
  \draw[thick] (1.9, -0.5) rectangle (2.9, 0.5);
  \draw[thick, dashed] (4.5, 1.4) -- (5, 1.4) node[left, xshift=-15pt] {$A: $};
  \draw[thick, double] (4.5, 1) -- (5, 1) node[left, xshift=-15pt] {$B: $};
  \draw[thick, dotted] (4.5, 0.6) -- (5, 0.6) node[left, xshift=-15pt] {$C: $};
  \draw[thick] (4.5, 0.2) -- (5, 0.2) node[left, xshift=-15pt] {$D: $};
  
\begin{scope}[every node/.style={circle,thick,draw}]
    \node(N) at (8 + 0, 0 +2) {$\emptyset$};
    
    \node(D) at (8 + -2, -1.25 +2) {$D$};
    
    \node(A) at (8 + 0, -1.25 +2) {$A$};
    \node[double](AB) at (8 + 0.75, -2.5 +2) {$B$};
    \node(AC) at (8 + -0.75, -2.5 +2) {$C$};
    \node(AD) at (8 + -2.25, -2.5 +2) {$D$};
    
    \node(B) at (8 + 2, -1.25 +2) {$B$};
    \node[double](BC) at (8 + 2.25, -2.5 +2) {$C$};
\end{scope}

\draw[thick] (N) -- (D) node[pos=0.7, above] {$1$};
\draw[thick] (N) -- (A) node[pos=0.5, left] {$2$};
\draw[thick] (N) -- (B) node[pos=0.7, above] {$4$};

\draw[thick] (A) -- (AB) node[pos=0.5, left] {$3$};
\draw[thick] (A) -- (AC) node[pos=0.5, left] {$2$};
\draw[thick] (A) -- (AD) node[pos=0.7, above] {$1$};

\draw[thick] (B) -- (BC) node[pos=0.6, left] {$1$};

\end{tikzpicture}
\caption{In this example, function $f$ counts the dots covered by a set of rectangles. On the right, the tree for stream $\sigma=(A,B,C,D)$ and $k=2$ is depicted. The labels
on the edges correspond to the increase in $f$. The
maximal leaves are highlighted.}
\label{fig:subm_example}
\end{figure}

\subsection{Overview of the Analysis}
For analyzing the algorithm, we will use a
sophisticated strategy to select one of the leaves and
only compare this leaf to the optimum.
We emphasize that this selection does not have to be
computed by the algorithm. In particular, it does not need to be computable by a streaming algorithm
and it can rely on knowledge of $\Eopt$ and $\Enoise$, which the algorithm does not have.
Since the algorithm always takes the best leaf,
we only need to
give a lower bound for one of them.
Before we describe this strategy,
we analyze the tree algorithm in two educational corner cases.

The first one shows that by a careful selection of a leaf the algorithm appears to
take elements based on the location of the $\Eopt$,
although it does not know them.
Let $r^\pi_i = \argmax_{e \in \sigma(-\infty,o^{\pi}_{1}]} f(e)$, that is, the most valuable element until the arrival of the first element from $\Eopt$.
Here $\argmax$ breaks ties in favor of the first element in $\sigma$.
We do not know when $o^{\pi}_1$ arrives, but we
know that the algorithm will have created a
node (with the root as its parent) for $r^\pi_1$ by then.
We define iteratively
$R^\pi_i = \{r^\pi_1,\dotsc,r^\pi_i\}$ and
$r^\pi_{i+1} = \argmax_{e \in \sigma(r^{\pi}_{i},o^{\pi}_{i+1}]} f(e \mid R^{\pi}_{i})$ for all $i$.
Again, we can be sure that $r^\pi_{i+1}$, which
yields the best increase for $R^\pi_i$ until
the arrival of $o^\pi_{i+1}$, is a appended to
the path $r^\pi_1\rightarrow\cdots \rightarrow r^\pi_i$.

This selection is inspired by the following idea.
Suppose we could partition the stream
into $k$ intervals such that in each
exactly one elements from $\Eopt$ appears.
Then a sensible approach would be to
start with an empty solution and 
greedily add the element that yields
the maximal increase
to our partial solution in each interval.
Clearly one such partition would be
$\sigma(o^\pi_i, o^\pi_{i+1}]$, $i=1,\dotsc,k$.
We note that while the selection above is similar, it does not
completely capture this.
Although $r^\pi_{i+1}$ is an element that
arrives before $o^\pi_{i+1}$, we cannot be certain
that it arrives after $o^\pi_i$. We only know
that it arrives after $r^\pi_i$.

Next, we prove that the solution $R^{\pi}_{k}$ is a $1/2$-approximation.
This already shows that the tree algorithm
is $1/2$-approximate even in the adversarial order model.
By definition of $R^\pi_i$ and $r^\pi_i$, we have
\begin{multline*}
    f(R^{\pi}_{k}) = \sum_{i=1}^{k} f(r^{\pi}_i \mid R^{\pi}_{i-1})
    \geq \sum_{i=1}^{k} f(o^{\pi}_i \mid R^{\pi}_{i-1}) \\
    = \sum_{i=1}^{k} [ f(o^{\pi}_i \mid R^{\pi}_{i-1}) - f(o^{\pi}_i \mid R^{\pi}_k) ]+ \sum_{i=1}^{k} f(o^{\pi}_i \mid R^{\pi}_k) .
\end{multline*}
Notice that due to submodularity the term
$f(o^{\pi}_i \mid R^{\pi}_{i-1}) - f(o^{\pi}_i \mid R^{\pi}_k)$ is always non-negative.
Moreover, if $o^{\pi}_i = r^{\pi}_i\in R^\pi_k$, it collapses to $f(o^{\pi}_i \mid R^{\pi}_{i-1})$.
Thus, we can bound the right term of the equation
and thereby $f(R^\pi_k)$ with
\begin{equation*}
    f(R^\pi_k) \ge \sum_{\substack{i=1 \\ r^{\pi}_i = o^{\pi}_i}}^{k} f(o^{\pi}_i \mid R^{\pi}_{i-1}) + \sum_{i=1}^{k} f(o^{\pi}_i \mid R^{\pi}_k) .
\end{equation*}
From submodularity and monotonicity of $f$ it follows that
\begin{equation*}
    \sum_{i=1}^{k} f(o^{\pi}_i \mid R^{\pi}_k)
    \ge f(O^\pi_k \mid R^\pi_k)
    = f(O^\pi_k\cup R^\pi_k) - f(R^\pi_k)
    \ge f(O^\pi_k) - f(R^\pi_k) .
\end{equation*}
Hence, we conclude that
\begin{equation*}
    2 f(R^\pi_k) \ge f(O^\pi_k) + \sum_{\substack{i=1 \\ r^{\pi}_i = o^{\pi}_i}}^{k} f(o^{\pi}_i \mid R^{\pi}_{i-1}) .
\end{equation*}
 This shows that $R^{\pi}_{k}$ is $1/2$-approximate, because $O^\pi_k = \Eopt$. Indeed, if a significant value of the elements in
 $\Eopt$ are taken, then $R^{\pi}_k$ is even better than $1/2$-approximate.
 
 Recall that the elements $\Eopt$ are ordered randomly
 in the adversarial-injections model.
 Hence, the worst-case in the analysis above is that
 $R^{\pi}_k$ is disjoint from $\Eopt$ for all realizations
 of $\pi$.  However, by a different analysis we can see that
 this case is in fact well-behaved.
 This is because the algorithm would select the
 same elements $r^\pi_1,\dotsc,r^\pi_k$ for every
 realization of $\pi$. Hence, we can safely drop
 the superscript $\pi$ in $R^\pi_i$ and $r^\pi_i$.
 Since for every element $o\in \Eopt$ there is some realization of
 $\pi$ where $o^\pi_i = o$, yet the algorithm does not pick $o^\pi_i$,
 we can bound the increase of each $r_i$ by
  \begin{equation*}
     f(r_i \mid R_{i-1})
     \ge \max_{o\in \Eopt} f(o \mid R_{i-1})
     \ge \frac 1 k \sum_{o\in \Eopt} f(o \mid R_{i-1}) .
  \end{equation*}
  By submodularity and monotonicity we get
  \begin{equation*}
     \frac 1 k \sum_{o\in \Eopt} f(o \mid R_{i-1})
     \ge \frac 1 k f(\Eopt \mid R_{i-1})
     \ge \frac 1 k (\OPT - f(R_{i-1})).
  \end{equation*}
This is the same recurrence formula as in the
classic greedy algorithm and by simple calculations
we get the closed form
\begin{equation*}
    f(R_k) \ge \left(1 - \left(1 - \frac 1 k\right)^k\right) \OPT
    \ge \left(1 - \frac 1 e\right) \OPT .
\end{equation*}
In other words, the algorithm is even 
$(1-1/e)$-approximate in this case.
In our main proof we will use a more involved
strategy for selecting a leaf.
This is to be able to combine the two approaches discussed above.
%We allow an element $o^\pi_i$ to be taken as $r^\pi_i$
%only if it gives us a good value compared
%to the expectation of an element picked uniformly
%at random from the remaining elements in $\Eopt$.
%Otherwise, we exclude $o^\pi_i$ from the choice
%of $r^\pi_i$. 
%Notice that not taking an opt element helps in ``saving'' the opt element for future. Otherwise, we just take the element that gives the maximum value. It may or may not be an opt element. But notice that if this is the case, then our solution should necessarily pick opt element some times as the value given by the non-opt element is not sufficient.  

\subsection{Analysis} \label{sec:submod-analysis}
Let us first define the selection of the leaf we
are going to analyze.
The elements on the path to this leaf will be
denoted by $s^{\pi}_1,\dotsc,s^{\pi}_k$ and
we write $S^\pi_i$ for $\{s^\pi_1,\dotsc s^\pi_i\}$. The elements are defined inductively,
but as opposed to the previous section we need
in addition indices $n_1,\dotsc,n_k$.
Recall, previously we defined the $(i+1)$'th element
$r^\pi_{i+1}$ as the best increase
in $\sigma(r^\pi_i, o^\pi_{i+1}]$.
Here, we use $n_{i+1}$ to describe the index of the
element from $\Eopt$ which constitutes the end of
this interval. It is not necessarily $o^\pi_{i+1}$
anymore. We always start with $n_1 = 1$, but 
based on different cases we either
set $n_{i+1} = n_i + 1$ or $n_{i+1} = n_i$.
We underline that $n_i$ is independent of the
realization of $\pi$.
In the following, $t\in [0, 1]$ denotes a parameter
that we will specify later.   

The element $s^{\pi}_i$ will be chosen
from two candidates $u^{\pi}_{i}$ and $v^{\pi}_{i}$.
The former is the best increase of elements
excluding $o^\pi_{n_i}$, that is,
\begin{equation*}
    u^{\pi}_{i} = \begin{cases}
    \argmax_{e \in \sigma(-\infty,o^{\pi}_{n_1})} f(e) &\text{ if $i=1$,} \\    
    \argmax_{e \in \sigma(s^\pi_i,o^{\pi}_{n_i})} f(e \mid S^\pi_{i-1}) &\text{ otherwise.}
    \end{cases}
\end{equation*}
The latter is defined in the same way, except
it includes $o^\pi_{n_i}$ in the choices, that is,
\begin{equation*}
    v^{\pi}_{i} = \begin{cases}
    \argmax_{e \in \sigma(-\infty,o^{\pi}_{n_1}]} f(e) &\text{ if $i=1$,} \\    
    \argmax_{e \in \sigma(s^\pi_{i-1},o^{\pi}_{n_i}]} f(e \mid S^\pi_{i-1}) &\text{ otherwise.}
    \end{cases}
\end{equation*}
We now define the choice of $s^{\pi}_i$ and $n_{i+1}$ based on the following two cases.
Note that the cases are independent from the realization of $\pi$.
\begin{description}
\item[Case 1: $\E_{\pi} f(u^{\pi}_{i} \mid S^{\pi}_{i-1}) \geq t \cdot \E_{\pi} f(o^{\pi}_{n_i} \mid S^{\pi}_{i-1})$.]
In this case, we set $s^{\pi}_i = u^{\pi}_{i}$ and $n_{i+1}=n_{i}$.
Notice that this means $s^\pi_i$ is chosen independently from $o^\pi_{n_i}$.
In other words, we did not see $o^\pi_{n_i}$, yet.
The element $o^\pi_{n_i}$ is still each of the remaining elements in $\Eopt$ with equal probability.
In the analysis this is beneficial, because
the distribution of $o^\pi_{n_i},\dotsc,o^\pi_{k}$
remains unchanged.
This is similar to the second case in the previous section.
\item[Case 2: $\E_{\pi} f(u^{\pi}_{i} \mid S^{\pi}_{i-1}) < t \cdot \E_{\pi} f(o^{\pi}_{n_i} \mid S^{\pi}_{i-1})$.]
Here, set $s^{\pi}_i = v^{\pi}_{i}$ and $n_{i+1}=n_{i}+1$. Now the distribution of $o^\pi_{i},\dotsc,o^\pi_k$ can
change. However, a considerable value of $s^\pi_i$ over different $\pi$ comes from taking $o^\pi_{n_i}$.
As indicated by the first case in the previous section
this will improve the guarantee of the algorithm.
\end{description}
The solution $S^\pi_k$ corresponds to a leaf in the tree algorithm.
Clearly, $u^\pi_1$ and $v^\pi_1$ are children of the root. Hence, $s^\pi_1$ is also a child. Then for induction we assume $s^\pi_i$
is a node, which implies $u^\pi_{i+1}$ and $v^\pi_{i+1}$ are also nodes:
The elements $u^\pi_{i+1}$ and $v^\pi_{i+1}$ are the first elements
after $s^\pi_i$ with the respective gains
($f(u^\pi_{i+1} \mid S^\pi_i)$ and $f(v^\pi_{i+1} \mid S^\pi_i)$). Hence, $s^\pi_{i+1}$ is a child of $s^\pi_i$.  

In order to bound $\E_\pi f(S^\pi_k)$, we will study more broadly
all values of $\E_\pi f(S^\pi_h)$ where $h\le k$.
To this end, we define a recursive formula $R(k,h)$
and prove that it bounds $\E_\pi f(S^\pi_h)/\OPT$ from below.
Then using basic calculus we will show that $R(k, k)\ge 0.5506$
for all $k$.
Initialize $R(k,0)=0$ for all $k$. 
Then let $R(k,h)$, $h \le k$, be defined by
\begin{equation*}
    R(k,h) = \min \left\{ \frac{t}{k} + \left(1 - \frac{t}{k} \right) R(k,h-1)\,,\;\;
    \frac{1}{k} + \left(1 - \frac{1+t}{k} \right) R(k-1,h-1)\,,\;\;
    \frac{1}{1+t} \right\}.
\end{equation*}
\begin{lemma}  \label{lem:sub-analysis}
For all instances of the problem and $h\le k$,
the solution $S^\pi_{h}$ as defined above satisfies
$\E_\pi f(S^\pi_h) \ge R(k,h) \OPT$ .
\end{lemma}

\begin{proof}
The proof is by induction over $h$.
For $h=0$, the statement holds
as $R(k,0) \OPT = 0 = \E_\pi f(S_0^\pi)$.
Let $h > 0$ and suppose the statement of the lemma holds
with $h-1$ for all instances of the problem.
Suppose we are given an instance with $k\ge h$.
We distinguish the two
cases $s^\pi_1=u^\pi_1$ and $s^\pi_1=v^\pi_1$. 

First, consider $\E_\pi f(u^{\pi}_{1}) \geq t \cdot \E_{\pi} f(o^{\pi}_{1})$,
which implies that $s^\pi_1=u^\pi_1$.
Note that $u^\pi_1$ is
the best element in $\sigma(-\infty,o^\pi_1)$,
consequently, its choice is independent from the
realization of $\pi$. Let us drop the superscript
in $u^\pi_1$ and $s^\pi_1$ for clarity.
We construct a new instance mimicking
the subtree of $s_1$.
Formally, our new instance still has the same
$k$ elements $\Eopt$, i.e., $k' = k$.
The stream is $\sigma' = \sigma(s^\pi_1,\infty)$ and,
  the submodular function $f':2^{U} \rightarrow \mathbb{R}$, $f'(T) \mapsto f(T \mid s_{1})$.
In this instance we have $\OPT'=f'(\Eopt)=f(\Eopt \mid s_1) \ge \OPT - f(s_1)$. It is easy to see that the elements
$s^{\prime\pi}_1,\dotsc,s^{\prime\pi}_{h-1}$ chosen in the new
instance correspond exactly
to the elements $s^\pi_2,\dotsc,s^\pi_h$.
Hence, with the induction hypothesis we get
\begin{equation*}
    \E_\pi f(S^\pi_h)
    = f(s_1) + \E_\pi f(S^{\pi}_{h} \mid s_1)
    = f(s_1) + \E_\pi f'(S^{\prime\pi}_{h-1})
    \ge f(s_1) + R(k, h-1) (\OPT - f(s_1)) .
\end{equation*}
By assumption we have $f(s_1) \ge t \cdot \E_\pi f(o^\pi_i) \geq t\cdot \OPT / k$. Together with $R(k, h-1)\le 1/(1+t) \le 1$
we calculate
\begin{equation*}
    f(s_1) + R(k, h-1) (\OPT - f(s_1))
    \ge \frac{t}{k} \OPT + R(k, h - 1)\left(1 - \frac t k\right) \OPT .
\end{equation*}
The right-hand side is by definition at least
$R(k, h) \OPT$.

Now we turn to the case
$\E_\pi f(u^{\pi}_{1}) < t \cdot \E_{\pi} f(o^{\pi}_{1})$,
which means $s^\pi_1 = v^\pi_1$ is chosen.
Similar to the previous case, we construct a new instance.
After taking $s^\pi_1$, our new instance has $k'=k-1$
elements $\Eopt' = \Eopt\setminus\{o^\pi_{1}\}$,
stream $\sigma' = \sigma(s_1,\infty)$, and
submodular function $f':2^E \rightarrow \mathbb{R}$, $f(T)\mapsto f(T \mid s^\pi_{1})$.
Thus, $\OPT'=f'(\Eopt') = f(\Eopt\setminus\{o^\pi_1\} \mid s^\pi_1) \ge \OPT - f(s^\pi_1 \cup o^\pi_1)$.
We remove $o^\pi_1$ from $\Eopt$,
because $s^\pi_1 = v^\pi_1$ depends on it.
The distribution of $o^\pi_2,\dotsc,o^\pi_k$
when conditioning on the value of $o^\pi_1$ (and
thereby the choice of $s^\pi_1$) is
still a uniformly random permutation of $\Eopt'$.
Like in the previous case, we can see that
$S^{\prime\pi}_{h-1} = S^\pi_h \setminus \{s^\pi_1\}$ and we can
apply the induction hypothesis. First, however,
let us examine $\E_\pi f(s^\pi_1\cup o^\pi_1)$.
Since we know that whenever $s^\pi_1 \neq o^\pi_1$ we
have $s^\pi_1 = u^\pi_1$, it follows that
\begin{equation*}
   \mathbb P_\pi [s^\pi_1 \neq o^\pi_1] \cdot \E_\pi [f(s_1) \mid s^\pi_1 \neq o^\pi_1]
    \le \E_\pi f(u^\pi_1)
    < t \cdot \E_\pi f(o^\pi_1) \le t \cdot \E_\pi f(s^\pi_1) .
\end{equation*}
Hence, we deduce
\begin{equation*}
   \E_\pi f(s^\pi_1\cup o^\pi_1)
   \le \E_\pi f(o^\pi_1) + \mathbb P_\pi[s^\pi_1 \neq o^\pi_1] \cdot \E_\pi [f(s_1) \mid s^\pi_1 \neq o^\pi_1] \le \E_\pi f(o^\pi_1) + t\cdot \E_\pi f(s^\pi_1) .
\end{equation*}
We are ready to prove the bound on
$\E_\pi f(S^\pi_h)$.
By induction hypothesis, we get
\begin{align*}
    \E_\pi f(S^\pi_h) &= \E_\pi f(s^\pi_1) + \E_\pi f'(S^{\prime\pi}_{h-1}) \\
    &\ge \E_\pi f(s^\pi_1) + R(k - 1, h - 1) (\OPT - \E_\pi f(s^\pi_1\cup o^\pi_1)) .
\end{align*}
Inserting the bound on $\E_\pi f(s^\pi_1\cup o^\pi_1)$
we know that the right-hand side is at least 
\begin{equation*}
    \E_\pi f(s^\pi_1) + R(k - 1, h - 1) (\OPT - \E_\pi f(o^\pi_1) - t\cdot \E_\pi f(s^\pi_1)) .
\end{equation*}
Using that $f(s^\pi_1) \ge f(o^\pi_1)$ for all $\pi$ and
$R(k - 1, h - 1) \cdot t \le t/(1+t) \le 1$ we bound
the previous term from below by
\begin{equation*}
    \E_\pi f(o^\pi_1) + R(k - 1, h - 1) (\OPT - (1 + t)\E_\pi f(o^\pi_1)) .
\end{equation*}
Finally, we use that $\E_\pi f(o_1^\pi) \ge \OPT / k$ and
$R(k - 1, h - 1) (1 + t) \le 1$ to arrive at
\begin{equation*}
    \frac 1 k \OPT + R(k - 1, h - 1) \left(\OPT - \frac{1 + t}{k}\OPT \right) \ge R(k, h) \OPT ,
\end{equation*}
which concludes the proof.
\end{proof}
With $t = 0.8$ we are able to show
that for sufficiently large $k$ the minimum in the definition of $R(k, k)$ is always attained by the first term.
Then, after calculating a lower bound on $R(k, k)$
for small values, we can easily derive a general bound.
\begin{lemma} \label{lem:submod-final}
With $t = 0.8$ for all positive integers k it holds
that $R(k,k)\ge 0.5506$ .
\end{lemma}
Figure~\ref{fig:recurrence} contains
a diagram (generated by computer calculation),
which shows that the formula tends to
a value between $0.5506$ and $0.5507$ for
$k\in \{0,\dotsc,10000\}$. The proof requires tedious and mechanical calculations and is thus deferred to Appendix~\ref{app:one}.
%and \textcolor{red}{hence is omitted here.} We refer the reader to Appendix~\ref{app:one} 
%\textcolor{red}{appendix of the arxiv version(cite stuff)} for complete %details.
 \begin{figure}
     \centering
     \includegraphics[width=\textwidth]{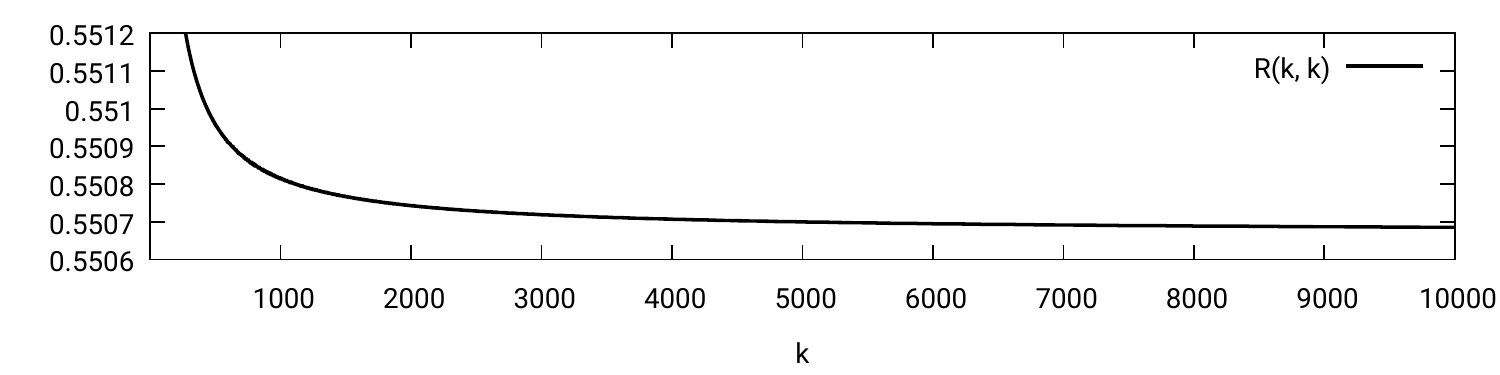}
     \caption{Values of the recurrence formula for $t = 0.8$.}
     \label{fig:recurrence}
 \end{figure}

%Arxiv version
%The proof requires tedious and mechanical calculations
%and is thus deferred to Appendix~\ref{app:one}.

%%% Local Variables:
%%% mode: latex
%%% TeX-master: "main_ai"
%%% End:

\section{Conclusion and Open Problems}\label{sec:conclusion}
In this paper, we introduced a semi-random model called adversarial-injections with the motivation of eliminating algorithms that overfit to random-order streams while still being easier than adversarial-order streams. We studied two classical problems in combinatorial optimization in this model.

For unweighted matching, we could beat $1/2$ in the streaming setting whereas we observed from \cite{gamlath2019online} that we could not beat $1/2$ in the online setting. This also makes our model non-trivial as there is a separation between the online and streaming setting.  

For monotone submodular maximization with cardinality constraint $k$, we obtained a $0.55$ approximation algorithm albeit with a huge memory footprint but importantly independent of $n$ (universe size). The obvious open question is whether one can design a $(1-1/e)$-approximation algorithm which stores number of elements that is independent of $n$. Does our algorithm have an approximation ratio of $1-1/e$? We observed that the algorithm in \cite{norouzi2018beyond} is a $1/2 + \eps$ approximation for a very small $\eps > 0$. The algorithm stores $\poly(k)$ elements. Can one design an algorithm that stores only $\poly(k)$ elements and beats $1/2$ by a significant constant or, even better, gets $1-1/e$? 

\bibliography{ref}
\appendix

%\begin{figure}
%    \centering
%    \includegraphics[width=\textwidth]{recursion}
%    \caption{Values of the recurrence formula for $t = 0.8$.}
%    \label{fig:recurrence}
%\end{figure}

\section{Analysis of the streaming algorithm for maximum matching}

We start by briefly recalling the algorithm \textsc{Match}. We refer the reader to Section \ref{alg:mat} for complete details. 

\textsc{Match} runs two algorithms in parallel and outputs the better solution. The first algorithm runs greedy whereas the second algorithm runs greedy up to certain point (Phase 1), and then starts collecting 3-augmenting paths. 

Now we recall some definitions introduced in Section~\ref{sec:matching} that we will use throughout the section. Let $M^*$ denote the maximum matching, $M_1$ the variable that is updated by the first algorithm and $M_2$ the matching produced by the second algorithm. Let $M^{(1)}_2$ denote this matching of the second algorithm just after Phase 1 ends.

\subsection{Removing assumption that $|M^*|$ is known} \label{app:matchopt}
Our algorithm \textsc{Match} assumed that it knew $|M^*|$. This is because the second algorithm that is run in \textsc{Match} needs to know $|M^*|$ as it starts collecting 3-augmenting paths when it has collected at least $|M^*|(1/2-\eps)$ edges. Until this point in the stream, it has collected exactly $|M_1|$ edges. By definition, $|M_1|$ is also a lower bound on $|M^*|$. Hence at any point, we will run multiple copies of the second algorithm initialized with a guess for $|M^*|$ based on what the value of $|M_1|$ is at that point. Thus, for any fixed $\delta > 0$, we can guess the value of $|M^*|$
up to a factor of $(1+\delta)$ by running the algorithm in parallel for all  powers $i$ of $1 + \delta$, that satisfy $|M_1|/(1+\delta) \leq (1+ \delta)^i \leq 4|M_1|/(1 - 2\eps) $. Notice that some copies of our algorithm may stop after some time as their $|M^*|$ estimate is no longer valid, others will continue and new ones with $|M^*|$ estimates that are already not running will start initialized with the matching $M_1$ at that point in the stream.
This increases algorithm's space only by a factor of  $O(\log \frac{4(1+\delta)}{1-2\eps})$
and deteriorates the solution value by at most $O(\delta |M^*|)$.

\subsection{Proof of Theorem \ref{thm:matching}}\label{app:match}

In this section, we prove that \textsc{Match} has an approximation ratio of at least $1/2+\gamma$ in expectation for some fixed constant $\gamma$ >0.

First we prove a lemma on the robustness
of the greedy algorithm
to small changes in the stream which is required to deal with correlations that may arise in the case where we need to show that the greedy algorithm picks a significant fraction of edges of $|M^*|$. For more on this, we refer the reader to Section \ref{sec:mat-overview}. 

\begin{lemma}\label{robust-greedy}
Let $\sigma$ and $\sigma'$ be streams of edges
in $G$ such that $\sigma$ can be transformed into $\sigma'$ by deleting an edge from $\sigma$. Let $M$ and $M'$ be the matchings computed
by the greedy algorithm on $\sigma$ and $\sigma'$ respectively. Then $| |M|-|M'| | \leq 1$.
\end{lemma}
\begin{proof}
Let $C = (M \setminus  M')\cup (M'\setminus M)$, that is, the symmetric difference of
$M$ and $M'$. Notice that $C$ is a collection of disjoint paths and cycles
that alternate between edges of $M$ and $M'$.
We claim that in this collection there is at most one path of non-zero length.
This implies the statement in the lemma.

We argue that if a path exists, it must contain the edge that was deleted. Hence, it
is the unique path.
Assume toward contradiction, that there exists a path in $C$, which does not contain the deleted edge.
We now closely examine the first edge $e$ of this path that arrives.
Note that this is the same for $\sigma$ and $\sigma'$.
In both runs of the greedy algorithm, the two vertices of $e$ are not incident to a matching edge
when $e$ arrives. This means that $e$ should have been taken in both runs and therefore cannot be in the
symmetric difference of the matchings, a contradiction.

\end{proof}
In our analysis we will use the following
lemma by Konrad et al which bounds the number of edges that cannot be augmented by 3-augmenting paths if size of maximal matching is small.

\begin{lemma}[Lemma~1 in~\cite{konrad2012maximum}]\label{non-3-aug} 

Let  $ \alpha > 0$, $M$
be a maximal matching in $G$ and $M^*$ be a
maximum matching in $G$ with $|M| \leq (1/2+\alpha)|M^*|$.
Then the number of 3-augmentable edges in $M$
is at least $(1/2 - 3 \alpha)|M^*|$.
In particular, the number of non-3-augmentable edges in $M$
is at most $4\alpha|M^*|$.
\end{lemma}
We are now ready to prove Theorem~\ref{thm:matching}.

\begin{proof}[Proof of Theorem~\ref{thm:matching}]
Define $\eps=1/50$, $\alpha = \eps$, and $\rho=\eps/4$. Without loss of generality, we assume that $|M^*| > 2/\rho$. Otherwise,
the algorithm can just store the whole graph as it is sparse i.e., it has linear number of edges.

Let $\sigma_1$ be the smallest prefix of the stream that contains
$k = \lceil \rho \cdot |M^*| \rceil$
elements of $M^*$. Further, let $\sigma_2$ be the remainder of the stream.
%Recall, the second algorithm first constructs a
%matching $|M'_2|$ greedily until
%it contains $(1/2-\eps)|M^*|$ edges (Phase~1).
Notice, that $M^{(1)}_2 \subseteq M_1$, since the first algorithm starts the
same way, but continues even after reaching this threshold.
Let $M^{(\sigma_1,1)}_2 \subseteq M_2^{(1)}$ be a random variable corresponding to only
those edges in $M^{(1)}_2$ taken during $\sigma_1$.
%Let $t^*$ the time when the second algorithm starts collecting 3-Aug paths and $t$ be the time just after the last element of $\sigma_1$ has appeared.
%We will analyze the events that the first phase of algorithm 2 ends
%in $\sigma_1$ and that it ends in $\sigma_2$ separately.
\begin{description}
\item[Case 1: For all permutations of $\Eopt$ it holds that  $M^{(\sigma_1,1)}_2 = M^{(1)}_2$.] 
Notice that by definition, $|M^{(\sigma_1,1)}_2| = |M^{(1)}_2| \ge (1/2-\eps) |M^*|$.
\begin{comment}
We recall that the first phase of the second algorithm stops,
once $|M^{(1)}_2| \ge (1/2-\eps) |M^*|$. Hence, one would expect
that in Case~1 the first phase is probably finished after $\sigma_1$.
We will show that because of the additional $2k$ in the bound, this
is in fact guaranteed. Although $\sigma_1$ depends on the random
order of $\Eopt$, those elements are only $k$ many in $\sigma_1$.
Hence, by adding and deleting up to $2k$ elements, we can transform
any two realizations of $\sigma_1$ into each other.
Since there must be at least one realization where
$|M^{(\sigma_1,1)}_2| \ge (1/2-\eps) |M^*| + 2k$, it follows from
Lemma~\ref{robust-greedy} that no matter which $\sigma_1$, it always holds
that
$|M^{(\sigma_1,1)}_2| \ge (1/2-\eps) |M^*| + 2k - 2k = (1/2-\eps) |M^*|$.
This implies that $M^{(1)}_2 = M^{(\sigma_1,1)}_2$.
\end{comment}
The basic idea for this case is to show
that in $\sigma_2$ there are
a lot of 3-augmenting paths that can be used to improve $M^{(1)}_2$ via \textsc{3-Aug-Paths}.

If $\mathbb |M_1| \ge |M^*|(1/2+\alpha)$, we are already
significantly better than $1/2$.
Hence, assume otherwise.
From Lemma~\ref{non-3-aug} it follows that
the number of non-3-augmentable edges
in $M_1$ is at most $4\alpha |M^*|$.
The number of 3-augmentable edges in $M^{(\sigma_1,1)}_2$ is obviously $|M^{(\sigma_1,1)}_2|$ minus the number non-3-augmentable edges in it. The former is at least $|M^*|(1/2-\eps)$
while the latter is a subset of the
non-3-augmentable edges in $M_1$ and
hence at most $4\alpha |M^*|$.
It follows that the number of 3-augmentable edges in $M^{(\sigma_1,1)}_2$ is
at least $(1/2-4\alpha - \eps)|M^*|$.

We will now restrict our attention to the subgraph of the edges in $M^{(\sigma_1,1)}_2$ and
$\sigma_2$ and the 3-augmentable edges there.
Recall, every 3-augmentable edge corresponds to a 3-augmenting path
that has two edges from $M^*$ and one from $M^{(\sigma_1,1)}_2$.
If at least one of the edges from $M^*$ appears
in $\sigma_1$, this edge is no longer 3-augmentable when we restrict ourselves
to $\sigma_2$.
However, by definition only $k$ of the edges from $M^*$ appear
in $\sigma_1$ and each of them can appear in only one 3-augmenting path.
In consequence, the number of 3-augmentable edges in $M^{(\sigma_1,1)}_2$ considering
only $\sigma_2$ is at least
\begin{equation*}
    \left(\frac 1 2 - 4\alpha - \eps\right)|M^*| - k \ge \left(\frac 1 2 - 4\alpha - \eps - \rho - \frac{1}{|M^*|}\right)|M^*| \ge \left(\frac 1 2 - 4\alpha - \eps - \frac{3\rho}{2} \right)|M^{*}| .
\end{equation*}
Here we use that by assumption $|M^*| > 2/\rho$.
After constructing the matching $M^{(1)}_2$, the second algorithm proceeds
to collect 3-augmenting paths for it using \textsc{3-Aug-Paths}.
We fix its parameter
\begin{equation*}
    \beta := \left(\frac 1 2 - 4\alpha-\eps- \frac{3\rho}{2} \right) / \left(\frac 1 2-\eps\right) .
\end{equation*}
Notice that after completing Phase~1, the second algorithm
has at least $\beta |M^{(1)}_2|$ many 3-augmenting paths in the remaining instance.
Hence, by Lemma~\ref{find-3-aug} we are guaranteed to find $(\beta^2/32)|M^{(1)}_2|$ many 3-augmenting paths. We conclude that
\begin{equation*}
  \mathbb |M_2| \ge \left(1 + \frac{\beta^2}{32}\right) |M^{(1)}_2| \ge \left(1 + \frac{\beta^2}{32}\right)(1 - \eps) |M^*| .
\end{equation*}
Using the definitions of the constants, we calculate that
$\beta = (4-43\eps)/(4-8\eps)$.
\item[Case 2: For at least one permutation of $\Eopt$ it holds that $M^{(\sigma_1,1)}_2 \subsetneq M^{(1)}_2$.]

Let $\sigma^*_1$ be the realization of the random variable $\sigma_1$ 
for such a permutation, i.e.,
we have
$|M^{(\sigma^*_1,1)}_2|  \leq (1/2-\eps) |M^*|$.
We will argue that in expectation (not just for $\sigma^*_1$) the greedy algorithm will select a
considerable number of elements from $M^*$.
This directly improves its guarantee:
Let $S = M^* \cap M_1$. Every
edge in $M^*\setminus S$ intersects with some edge in $M_1\setminus S$, but every
edge in $M_1\setminus S$ can only intersect with two edges in $M^*\setminus S$. This implies
$2|M_1\setminus S|\ge|M^*\setminus S|$ and consequently 
\begin{align}
   |M_1| \ge (|M^*|+|S|)/2.  \label{eq-match:eq1}
\end{align}

We bound $\mathbb E[|S|]$ from below by examining the elements of $M^*$ in $\sigma_1$, denoted by $e^*_1 , e^*_2, \dotsc, e^*_k$.
Let $\sigma^{(1)},\dotsc,\sigma^{(k)}$ correspond to
the prefix of $\sigma$ until right before $e^*_1,\dotsc,e^*_k$ arrive.
Further, define $M_1^{(1)},\dotsc,M_1^{(k)}$ as the value
of $M_1$ after each prefix $\sigma^{(1)},\dotsc,\sigma^{(k)}$.

Notice that $\sigma^{(k)}$ can be transformed
to $\sigma^*_1$
by adding and deleting at most $2k$ elements.
Thus, it follows from Lemma~\ref{robust-greedy} that for all $i\le k$ and any $\sigma_1$,
\begin{equation*}
    |M^{(i)}_1| \le |M^{(k)}_1| \le (1/2-\eps)|M^*| + 2k .
\end{equation*}
This implies that the number of edges in $M^*$ not intersecting
with edges in $M^{(i)}_1$ is at least $|M^*| - 2|M^{(i)}_1| \ge 2\eps |M^*| - 4k$.
If $e^*_i$ is one of these edges, then it is taken in $M_1$.
The probability for each element in $M^*$, which has not arrived yet,
to be $e^*_i$ is equal.
Hence, conditioning on some choice of $\sigma_1^{(i)}$ the probability
that $e^*_i$ is taken in $M_1$ is at least
$(2\eps |M^*| - 4k) / (|M^*| - (i-1)) \ge 2\eps - 4k/|M^*|$. 
Thus,
\begin{multline*}
\E[S] \ge \sum_{i = 1}^k \mathbb P[e^*_i \in M_1]
     \ge k \left(2\eps- \frac{4 k}{|M^*|}\right) \\
     \ge \rho |M^*| \left(2\eps- 4\frac{\rho |M^*| + 1}{|M^*|} \right)
     = \left(2\eps\rho - 4\rho^2 - \frac{4\rho}{|M^*|}\right) |M^*| .
\end{multline*}
With assumption $M^* > 2/\rho$, (\ref{eq-match:eq1}) we conclude that
\begin{equation*}
    \E[|M_1|] \geq \left(\frac 1 2+2 \eps\rho - 6\rho^2\right)|M^*| = \left(\frac{1}{2} + \frac{\eps^2}{8}\right) |M^*| .
\end{equation*}
\end{description}
Taking the worst of the bounds, we calculate the constant
\begin{equation*}
    \gamma = \min\left\{\eps, \frac{1}{32}\left(\frac{4 - 43\eps}{4 - 8\eps}\right)^2 (1 - \eps) -  \eps,
    \frac{\eps^2}{8}\right\} = \frac{1}{20000} . \qedhere
\end{equation*}
\end{proof}

\section{Omitted proofs for submodular function maximization}
\label{app:one}

We start by briefly recalling the algorithm. We refer the reader to Section~\ref{alg:submod} for the complete details. 

The algorithm produces a tree of height at most $k$ (cardinality constraint) and each root to node (at height $i$) path corresponds to a solution (of $i$ elements). Each node has at most $|I|$( set of all possible increments) children.  At the beginning of the stream, the root of the tree stores the empty set. For each element $e$ in the stream, we add it as a child of any node $T$ at height less than $k$, if for no existing child $c$ of $T$ it holds that 
$f(c \mid S )$=$f(e \mid S )$ where $S$ denotes the solution corresponding to the path from root to $T$. At the end of the stream, the algorithm produces the best solution among all leaves. 

We will choose $I$ to be a set of size $O(k)$. For this however, we first need to know $\OPT$. We show below how to remove this assumption with a small increase in space.

\subsection{Assumption that $\OPT$ is known} \label{app:two}
The algorithm presented in Section~\ref{alg:submod}
uses the assumption that it knows the value $\OPT$. We will describe
in the following how to remove this assumption.
We use the same trick as described in~\cite{badanidiyuru2014streaming}.
Let $\delta > 0$ be a small constant. 
It is easy to see that when using some value $g \le \OPT$
instead of $\OPT$, the algorithm produces
a solution of value at least $0.5506 g$.
We will run the algorithm in parallel for
multiple guesses of $g$.
If $g\le \OPT \le (1 + \delta)g$ for some guess,
we would only loose a factor of $(1+\delta)$
in the approximation ratio.
However, the range in which $\OPT$ lies cannot be bounded.
Hence, we must adapt our guesses as we read the stream.
To that end, we keep track of the maximum element in the stream at any point of time.
Let $m_i$ denote the maximum element after observing the first $i$ elements of the stream $\sigma$, i.e., $m_i = \max_{j \le i} f( \{ \sigma_j \})$.
It is easy to see that any subset of $\sigma_1,\dotsc,\sigma_i$ with cardinality at most $k$
has a value between $m_i$ and $k \cdot m_i$.
Let
\begin{equation*}
    G_i = \left\{ (1+\delta)^{j}: j \in \mathbb{Z}, \frac{1}{1+\delta} m_i \leq (1+\delta)^{j} \leq \frac{k}{\delta} m_i \right\} .
\end{equation*}
At any point $i$ in the stream, we will run
our algorithm in parallel for all guesses in $G_i$.
When a new maximum element arrives,
this may remove previously existing
guesses, in which case we stop the execution for this guess
and dismiss its result.
On the other hand, new guesses may be added and
we start a new execution pretending the stream begins at
$\sigma_i$.

Let us consider $g$, the correct guess for $\OPT$, i.e.,
$g \le \OPT \le (1 + \delta) g$. Once $g$ is added to $G_i$,
it remains in the set of guesses until the end of execution:
If it was removed, this would mean it exists an element
of value greater than $(1 + \delta) g \ge \OPT$. However, no set smaller
than $k$ can have a value larger than $\OPT$.
It remains to check that the error induced by starting
the algorithm late is not significant.
Let $O$ denote the elements from the optimal solution
and $O'\subseteq O$ those that arrived before execution
for $g$ was started.
All elements in $O'$ were smaller than $g \delta / k$.
Hence, $f(O') \le k\cdot g\delta/k = g\delta$.
This implies
\begin{equation*}
    f(O\setminus O') \ge f(O) - f(O') \ge \OPT - \delta g
    \ge (1 - \delta) \OPT .
\end{equation*}
Hence, the approximation ratio decreases by a factor of at most $(1 - \delta)/(1 + \delta)$
and the space increases by a factor of
$\log_{1+\delta}(k (1 + \delta) / \delta) = O(1/\delta \log(k/\delta))$.

\subsection{Bounding the number of increases}
\label{app:three}

Recall, the tree algorithm stores roughly $|I|^k$
elements. Here $I = \{f(e \mid S) \mid S\subseteq E, |S| < k, e\in E\}$
is the set of all possible increases of $f$.
In order to achieve a reasonable memory bound, we need
to make sure that $|I|$ is small.
In the following we describe a way 
that bounds $|I|$ by $O(k)$ and only decreases the
approximation ratio marginally.

We assume that $\OPT$ is known (see previous section). Let $\delta > 0$ be a small constant. We divide the range $0$ to $\OPT$ into $k/\delta$ buckets each of size $\delta \cdot \OPT/k$. The idea is that $I$ now represents the set of possible range of increases of $f$ where each bucket corresponds to a range.  We now argue that this discretization does not affect the approximation ratio much. Recall that the recursion $R(k,h)$ was defined as follows:
\begin{equation*}
    R(k,h) = \min \left( \frac{t}{k} + \left(1 - \frac{t}{k} \right) R(k,h-1) ,
    \frac{1}{k} + \left(1 - \frac{1+t}{k} \right) R(k-1,h-1),
    \frac{1}{1+t} \right).
\end{equation*}
In Section~\ref{sec:submod-analysis}, $R(k,k) \cdot \OPT$ was proven to be a lower bound on the expected value of the solution returned by the algorithm. Due to bucketing, the element that is picked now might differ in value by at most  $(\OPT \cdot \delta/k)$ from the value promised by the analysis. As the recursion $R(k,k)$ has depth $k$ (the solution $S_k$ consists of $k$ elements), it is not hard to see that one can lower bound the expected value of the returned solution after bucketing by $R(k,k) \cdot \OPT - k \cdot O(\delta/k) \OPT$. Thus the loss incurred by discretization can be made arbitrarily small by appropriately selecting $\delta>0$.

\subsection{Analysis of recursion function}
In order to prove that our algorithm is $0.55$ competitive, we will lower bound the value of the recursion $R(k,h)$ where $R(k,h)$ was defined as the lower bound on the approximation ratio of the solution (we defined in Section~\ref{sec:submod-analysis}) at height $i$ as compared to $\OPT$ over all submodular functions, streams and opt elements. The heart of the proof lies in the fact that for a certain threshold $t$ and large $k$, the complicated recursion which involves taking the minimum of three terms simplifies to a recursion (that one can solve easily) involving just the first term. We prove this in Lemma \ref{lem:lemma2}. For complete details on the recursion definition and the solution we analyze, we refer the reader to Section \ref{sec:submod-analysis}. 

We first state some technical claims which would be helpful in proving Lemma \ref{lem:lemma2}.

\begin{claim} \label{claim:cl1}
The following inequality is satisfied for all $x \in [-0.1,0]$:
$$ e^x - \frac{x^2}{2} \leq 1 + x \leq e^x - \frac{x^2}{2} - \frac{x^3}{6}. $$
\end{claim}

\begin{proof}
We first write down the taylor expansion of $e^x$:
$$e^x = \sum_{i=0}^{\infty} \frac{x^i}{i!}.$$ 
For $x \in [-0.1,0]$, 
$$\sum_{i=3}^{\infty} \frac{x^i}{i!} \leq -\frac{|x|^3}{6} + \frac{|x|^3}{24}  \cdot \sum_{i=1}^{\infty} |x|^i = -\frac{|x|^3}{6} + \frac{|x|^3}{24} \cdot \frac{1}{1-|x|} \leq 0.$$
Hence, $e^x \leq 1 + x + \frac{x^2}{2}.$ \\
Similarly, 
$$\sum_{i=4}^{\infty} \frac{x^i}{i!} \geq \frac{|x|^4}{24} - \frac{|x|^4}{120}  \cdot \sum_{i=1}^{\infty} |x|^i = \frac{|x|^4}{24} - \frac{x^4}{120} \cdot \frac{1}{1-|x|} \geq 0.$$
Hence, $e^x \geq 1 + x + \frac{x^2}{2} + \frac{x^3}{6}.$

\end{proof}

\begin{claim} \label{claim:cl2}
For $k \geq 999$ and $t\leq 1$, the following inequality is satisfied:
$$\frac{1}{2}\sum_{i=2}^{\infty} \left( \frac{t^2 \cdot e^{\frac{t}{k}} }{1.9 \cdot k} \right)^i \leq \frac{t^4}{3 \cdot k^2}.$$
\end{claim}

\begin{proof}
By using formula for sum of an infinite geometric series:
\begin{align*}
    &\frac{1}{2}\sum_{i=2}^{\infty} \left( \frac{t^2 \cdot e^{\frac{t}{k}} }{1.9 \cdot k} \right)^i = \frac{t^4 \cdot e^{\frac{2 \cdot t}{k}} }{1.9^2 \cdot k^2} \cdot \frac{1}{1 - \frac{t^2 \cdot e^{\frac{t}{k}} }{1.9 \cdot k}} \leq \frac{t^4}{3 \cdot k^2}.
\end{align*}
\end{proof}
Now we prove a closed form expression for $R(k,h)$ when
only the first rule is applied.
\begin{claim}\label{claim:exact-exp}
For any integer $k\geq 1000$ and any non-negative integer $h \leq k$ such that $R(k, h') = t/k + (1 - t/k) R(k, h'-1)$ for all $1 \le h' \le h$ it holds that
$$
 R(k,h) = 1 - \left( 1 - \frac{t}{k}\right)^h.
$$
\end{claim}
In fact, Lemma~\ref{lem:lemma2} will show that
for large values of $k$ (and any $h$) the condition
and thereby this closed form holds.

\begin{proof}
We will prove the equality by induction on $h$ for any $k \geq 1000$.  \\
For $h=0$, the equality holds as $R(k,h)=0$.
Then by induction hypothesis
\begin{align*}
    R(k,h) &= \frac{t}{k} + \left( 1 - ( 1 - \frac{t}{k})^{h-1} \right) \cdot (1 - \frac{t}{k}) \\
    &= \frac{t}{k} + 1 - \frac{t}{k} -\left( 1 - \frac{t}{k}\right)^h \\
    &= 1 - \left( 1 - \frac{t}{k}\right)^h. 
\end{align*}
\end{proof}

\begin{lemma}\label{lem:lemma2} 
With $t = 0.8$ for every $k\ge 1000$ and $h \leq k$ it holds
that
\begin{equation*}
    R(k,h) = \frac{t}{k} + \left(1 - \frac{t}{k} \right) R(k,h-1) .
\end{equation*} 
\end{lemma}

\begin{proof}
We will prove by induction on $k$ and $h$.\\
For $h=1$, $R(k,1) = \frac{t}{k}$. For $k=1000$, we have verified that the induction hypothesis holds by computer assisted calculation. Hence the base case holds. \\ Recall that $R(k,h)$ is defined as: 
$$  R(k,h) = \min \left( \frac{t}{k} + R(k,h-1) \cdot \left(1 - \frac{t}{k} \right) ,   \frac{1}{k} + R(k-1,h-1) \cdot \left(1 - \frac{1+t}{k} \right)  , \frac{1}{1+t} \right).
$$

By induction hypothesis and Claim \ref{claim:exact-exp}, we know that $R(k,h-1) = 1 - (1 -t/k)^{h-1}$. Hence for $k \geq 1000$ and $h \leq k$, we get:
\begin{align*}
    \frac{t}{k} + R(k,h-1) \cdot \left(1 - \frac{t}{k} \right)  &=\frac{t}{k} + (1 - (1 -\frac{t}{k})^{h-1})\cdot \left(1 - \frac{t}{k} \right) \\
    &= 1 - (1 -\frac{t}{k})^{h} \\
    &\leq 1 - (1 - \frac{t}{k})^{k}. \\
    \intertext{As $(1 - \frac{t}{k})^{k}$ is monotonically decreasing in $k$ and $t=0.8$, we get:}
    \frac{t}{k} + R(k,h-1) \cdot \left(1 - \frac{t}{k} \right)  &\leq 1 - (1 - \frac{t}{1000})^{1000}  \\
    &\leq 0.5509 \\ 
    &\leq \frac{1}{1+t}. 
\end{align*}
Hence it suffices to prove the below for $k \geq 1000$ and $h \leq k$: 
\begin{align}
\frac{t}{k} + R(k,h-1) \cdot \left(1 - \frac{t}{k} \right) &\leq  \frac{1}{k} + R(k-1,h-1) \cdot \left(1 - \frac{1+t}{k} \right). \nonumber 
\end{align}

Or, equivalently:
\begin{align}
t &\leq 1 + R(k-1,h-1) \cdot \left(k - 1 - t \right) - R(k,h-1) \cdot \left(k - t \right). \nonumber \\
\intertext{By induction hypothesis and Claim \ref{claim:exact-exp}, we know that $R(k,h-1) = 1 - (1 -\frac{t}{k})^{h-1}$ and $R(k-1,h-1) = 1 - (1 - \frac{t}{k-1})^{h-1}$.}
t&\leq 1 + (1 - (1 - \frac{t}{k-1})^{h-1})  \cdot \left(k - 1 - t \right) - ( 1 - (1 - \frac{t}{k})^{h-1} )   \cdot \left(k - t \right) \nonumber \\
&\leq 1 + \left( k-1-t - \frac{(k-1-t)^h}{(k-1)^{h-1}} \right) - \left( k -t - \frac{(k-t)^h}{(k)^{h-1}} \right) \nonumber\\
&\leq \frac{(k-t)^h}{(k)^{h-1}} - \frac{(k-1-t)^h}{(k-1)^{h-1}} \nonumber \\
&\leq \underbrace{ k \cdot (1 - \frac{t}{k})^h - (k-1) \cdot (1 - \frac{t}{k-1})^{h}}_{E_1}.  \label{eq:eq1} 
\end{align}

$E_1$ is monotonically decreasing in $h$ as shown below:

\begin{align*}
    \left( k \cdot (1 - \frac{t}{k})^h - (k-1) \cdot (1 - \frac{t}{k-1})^{h} \right) - \left( k \cdot (1 - \frac{t}{k})^{h+1} - (k-1) \cdot (1 - \frac{t}{k-1})^{h+1} \right) &\geq 0 \\
    k \cdot (1 - \frac{t}{k})^h \cdot (1 - (1 - \frac{t}{k})) - (k-1) \cdot (1 - \frac{t}{k-1})^{h} \cdot (1 - (1 - \frac{t}{k-1})) &\geq 0\\
    t \cdot (1 - \frac{t}{k})^h  - t \cdot (1 - \frac{t}{k-1})^{h} &\geq 0.
\end{align*}

Hence we can lower bound $E_1$ by assuming $h=k$. We lower bound $E_1$ below:  

\begin{align}
    E_1 &\geq k \cdot (1 - \frac{t}{k})^k - (k-1) \cdot (1 - \frac{t}{k-1})^{k}. \nonumber \\
    \intertext{By using Claim \ref{claim:cl1}, $k \geq 1000$ and $t \leq 1$, we get:}
    E_1 &\geq k \cdot (e^{\frac{-t}{k}} - \frac{t^2}{2\cdot k^2} ) ^ {k}  -
    (k-1) \cdot \left( \underbrace{e^{\frac{-t}{k-1}} - \frac{t^2}{2\cdot (k-1)^2} + \frac{t^3 }{6 \cdot (k-1)^3}}_{T_1} \right) ^ {k} \nonumber \\
    &\geq k \cdot e^{-t} \cdot \underbrace{\left( 1 - \frac{t^2 \cdot e^{\frac{t}{k}} }{2 \cdot k^2} \right)^{k} }_{T_2} -  (k-1)\cdot T_1 \cdot e^{-t} \cdot \underbrace{ \left(1 - \frac{t^2 \cdot e^{\frac{t}{k-1}}}{2\cdot (k-1)^2} + \frac{t^3 \cdot e^{\frac{t}{k-1}}}{6 \cdot (k-1)^3} \right)^{k-1} }_{T_3}.  \label{eq:eq2}
\end{align}

\begin{comment}
Using the fact that $e^x \approx 1+x$ , we get:(\textbf{Cheating here. Have to fix later.})
\begin{align*}
t &\leq h \cdot e^{ -\frac{t \cdot h}{h}} - (h-1) \cdot e^{-\frac{t \cdot h}{h-1}} \\
&\leq h \cdot e^{-t} - (h-1) \cdot e^{-t} \cdot e^{\frac{-t}{h-1}} \\
&\leq h \cdot e^{-t} - (h-1) \cdot e^{-t} \cdot (1 - \frac{t}{h-1} ) \\
&\leq h \cdot e^{-t} - e^{-t} \cdot (h -1 - t) \\
&\leq e^{-t} \cdot (1+t)
\end{align*}
\end{comment}

We lower bound the term $T_2$. 
\begin{align*}
T_2 &= \left( 1 - \frac{t^2 \cdot e^{\frac{t}{k}} }{2 \cdot k^2}  \right)^{k}\\
&\geq 1 - \frac{t^2 \cdot e^{\frac{t}{k}} }{2 \cdot k}  - \sum_{i=2}^{\infty} \frac{k^{i}}{2} \cdot \left( \frac{t^2 \cdot e^{\frac{t}{k}} }{2 \cdot k^2} \right)^i \text{ (By Binomial Expansion.) }\\
&= 1 - \frac{t^2 \cdot e^{\frac{t}{k}} }{2 \cdot k}  - \frac{1}{2}\sum_{i=2}^{\infty} \left( \frac{t^2 \cdot e^{\frac{t}{k}} }{2 \cdot k} \right)^i \\
&\geq 1 - \frac{t^2 \cdot e^{\frac{t}{k}} }{2 \cdot k} - \frac{t^4}{3 \cdot k^2}.  \text{ (By Claim \ref{claim:cl2}, $t \leq 1$ and $k \geq 1000.$) }
\end{align*}

We now simplify and lower bound $k \cdot e^{-t} \cdot T_2$:

\begin{align}
    k \cdot e^{-t} \cdot T_2  &\geq   k \cdot e^{-t} \cdot  (1 - \frac{t^2 \cdot e^{\frac{t}{k}} }{2 \cdot k} - \frac{t^4}{3 \cdot k^2}) \nonumber \\
    &= k \cdot e^{-t} - \frac{t^2 \cdot e^{\frac{t}{k}} \cdot e^{-t}}{2} - \frac{t^4 \cdot e^{-t}}{3 \cdot k}. \label{eq:eq3}
\end{align}

We now upper bound $T_3$. 
\begin{align*}
    T_3 &= \left(1 - \frac{t^2 \cdot e^{\frac{t}{k-1}}}{2\cdot (k-1)^2} + \frac{t^3 \cdot e^{\frac{t}{k-1}}}{6 \cdot (k-1)^3} \right) ^ {k-1}. \\
    \intertext{By using Binomial Expansion, we get:}
    T_3&\leq 1 - \frac{t^2 \cdot e^{\frac{t}{k-1}}}{2\cdot (k-1)} + \frac{t^3 \cdot e^{\frac{t}{k-1}}}{6 \cdot (k-1)^3} +  \sum_{i=2}^{\infty} \frac{(k-1)^{i}}{2} \cdot \left( \frac{t^2 \cdot e^{\frac{t}{k-1}}}{2\cdot (k-1)^2} + \frac{t^3 \cdot e^{\frac{t}{k-1}}}{6 \cdot (k-1)^3}\right)^i \\
    &= 1 - \frac{t^2 \cdot e^{\frac{t}{k-1}}}{2\cdot (k-1)} + \frac{t^3 \cdot e^{\frac{t}{k-1}}}{6 \cdot (k-1)^3} +  \frac{1}{2} \cdot \sum_{i=2}^{\infty} \left( \frac{t^2 \cdot e^{\frac{t}{k-1}}}{2\cdot (k-1)} + \frac{t^3 \cdot e^{\frac{t}{k-1}}}{6 \cdot (k-1)^2}\right)^i. \\
    \intertext{By using that $t\leq 1$ and $k \geq 1000$, we get:}
    T_3 &\leq 1 - \frac{t^2 \cdot e^{\frac{t}{k-1}}}{2\cdot (k-1)} + \frac{t^3}{3 \cdot (k-1)^2} + \frac{1}{2} \cdot \sum_{i=2}^{\infty} \left( \frac{t^2 \cdot e^{\frac{t}{k-1}}}{1.9 \cdot (k-1)} \right)^i\\
    &\leq 1 - \frac{t^2 \cdot e^{\frac{t}{k-1}}}{2\cdot (k-1)} + \frac{t^3}{3 \cdot (k-1)^2} + \frac{t^4}{3 \cdot (k-1)^2} \text{ (By Claim \ref{claim:cl2}, $t \leq 1$ and $k \geq 1000.$) } \\
    &\leq 1 - \frac{t^2 \cdot e^{\frac{t}{k-1}}}{2\cdot (k-1)} + \frac{2 \cdot t^3}{3 \cdot (k-1)^2}.
\end{align*} \\

We now upper bound $T_1$. By using Claim \ref{claim:cl1}, $k \geq 1000$ and $t\leq 1$, we get:
$$
T_1 \leq 1 - \frac{t}{k-1} + \frac{t^3}{6 \cdot (k-1)^3}.
$$

We now simplify and upper bound $(k-1)\cdot T_1 \cdot e^{-t} \cdot T_3$:
\begin{align}
    & (k-1)\cdot T_1 \cdot e^{-t} \cdot T_3 \nonumber \\
    &\leq (k-1) \cdot e^{-t} \cdot T_1 \cdot (1 - \frac{t^2 \cdot e^{\frac{t}{k-1}}}{2\cdot (k-1)} + \frac{2 \cdot t^3}{3 \cdot (k-1)^2}) \nonumber \\
    &= e^{-t} \cdot T_1 \cdot ( (k-1) - \frac{t^2 \cdot e^{\frac{t}{k-1}}}{2} + \frac{2 \cdot t^3}{3 \cdot (k-1)}) \nonumber \\
    &\leq e^{-t} \cdot (1 - \frac{t}{k-1} + \frac{t^3}{6 \cdot (k-1)^3}) \cdot ( (k-1) - \frac{t^2 \cdot e^{\frac{t}{k-1}}}{2} + \frac{2 \cdot t^3 }{3 \cdot (k-1)}) \nonumber \\
    &\leq e^{-t} \cdot\left( \left((k-1) - \frac{t^2 \cdot e^{\frac{t}{k-1}}}{2} + \frac{2 \cdot t^3}{3 \cdot (k-1)}\right) + \left(-t + \frac{t^3 \cdot e^{\frac{t}{k-1}}}{2 \cdot (k-1)}\right) + \left(\frac{t^2 \cdot e^\frac{t}{k-1}}{3 \cdot (k-1)^2}\right) \right) \nonumber \\
    &\leq e^{-t} \cdot\left( k-1 - t -\frac{t^2 \cdot e^{\frac{t}{k-1}}}{2} + \frac{7 \cdot t^3 \cdot e^{\frac{t}{k-1}}}{6 \cdot (k-1)}  + \frac{t^2 \cdot e^\frac{t}{k-1}}{3 \cdot (k-1)^2} \right). \label{eq:eq4}
\end{align}

By replacing (\ref{eq:eq3}) and (\ref{eq:eq4}) in (\ref{eq:eq2}), we get: 

\begin{align}
    E_1 &\geq k \cdot e^{-t} - \frac{t^2 \cdot e^{\frac{t}{k}} \cdot e^{-t}}{2} - \frac{t^4 \cdot e^{-t}}{3 \cdot k} - e^{-t} \cdot( k-1 - t -\frac{t^2 \cdot e^{\frac{t}{k-1}}}{2} + \frac{7 \cdot t^3 \cdot e^{\frac{t}{k-1}}}{6 \cdot (k-1)}  + \frac{t^2 \cdot e^\frac{t}{k-1}}{3 \cdot (k-1)^2})   \nonumber \\
    &= e^{-t}\cdot(1+t) + e^{-t}\cdot( \frac{t^2 \cdot e^{\frac{t}{k-1}}}{2}  -\frac{t^2 \cdot e^{\frac{t}{k}}}{2} ) -e^{-t}\cdot(\frac{t^4}{3 \cdot k} + \frac{7 \cdot t^3 \cdot e^{\frac{t}{k-1}}}{6 \cdot (k-1)}) -e^{-t}\cdot \frac{t^2 \cdot e^\frac{t}{k-1}}{3 \cdot (k-1)^2} \nonumber \\
    &\geq e^{-t}\cdot(1+t)  -e^{-t}\cdot\frac{2 \cdot t^3 \cdot e^{\frac{t}{k-1}}}{(k-1)} -e^{-t}\cdot \frac{t^2 \cdot e^\frac{t}{k-1} }{3 \cdot (k-1)^2}. \label{eq:eq5}
\end{align}

By using (\ref{eq:eq1}) and (\ref{eq:eq5}), we get that it suffices to prove the following:

\begin{align*}
    t \leq e^{-t}\cdot(1+t)  -e^{-t}\cdot\frac{2 \cdot t^3 \cdot e^{\frac{t}{k-1}}}{(k-1)} -e^{-t}\cdot \frac{t^2 \cdot e^\frac{t}{k-1} }{3 \cdot (k-1)^2}.
\end{align*}

Here, $t=0.8$ satisfies the final inequality. 
\end{proof}

\begin{customthm}{8}
For all positive integers k, $R(k,k)\geq 0.5506.$
\end{customthm}

\begin{proof}
For $k \leq 1000$, $R(k,k) \geq 0.5506$ can easily be verified with computer assistance, since
it involves only a dynamic program of size $1000\times 1000$.
For $k > 1000$,  by Lemma \ref{lem:lemma2} and Claim \ref{claim:exact-exp}, $R(k,k) \geq 1- (1 - \frac{0.8}{k})^k \geq 1- e^{-0.8} \geq 0.5506$.
\end{proof}

%\appendix

\end{document}